\documentclass[amssymb,pra,twocolumn,notitlepage,superscriptaddress]{revtex4-2}
\pdfoutput=1
\usepackage{times}
\usepackage{amsmath}
\usepackage{ulem}
\usepackage{graphicx}
\usepackage{dsfont}
\usepackage{color}

\usepackage{braket}
\usepackage[colorlinks=true, letterpaper=ture, pdfstartview=FitV, linkcolor=blue, citecolor=blue, urlcolor=blue]{hyperref}
\usepackage{amssymb}
\usepackage{enumitem}
\usepackage{mathtools}
\usepackage{mathrsfs}
\usepackage[most]{tcolorbox}
\usepackage{amsthm}
\newtheorem{lemma}{Lemma}
\begin{document}


\title{Probing Quantum Information Scrambling via Local Randomized Measurements}

\author{Yan-Ming Chen}
 \affiliation{Key Laboratory of Atomic and Subatomic Structure and Quantum Control (Ministry of Education), Guangdong Basic Research Center of Excellence for Structure and Fundamental Interactions of Matter, and School of Physics, South China Normal University, Guangzhou 510006, China}

\author{Dan-Bo Zhang}
 \altaffiliation[Contact author: ]{dbzhang@m.scnu.edu.cn}
 \affiliation{Key Laboratory of Atomic and Subatomic Structure and Quantum Control (Ministry of Education), Guangdong Basic Research Center of Excellence for Structure and Fundamental Interactions of Matter, and School of Physics, South China Normal University, Guangzhou 510006, China}
 \affiliation{Guangdong Provincial Key Laboratory of Quantum Engineering and Quantum Materials, Guangdong-Hong Kong Joint Laboratory of Quantum Matter, and Frontier Research Institute for Physics, South China Normal University, Guangzhou 510006, China}

\date{\today}

\begin{abstract}
In quantum many-body dynamics, locally encoded information typically scrambles across the entire system, becoming inaccessible to local probes. The upper bound of accessible information of local probes can be characterized by the Holevo information via optimal measurement. In this work, we investigate the information dynamics of quantum scrambling utilizing local randomized probes, quantified by the averaged accessible information (AAI). We derive an analytical expression for the AAI under Haar-random measurements and demonstrate that it is a function of purity of local reduced density matrix. Operationally, we employ the classical shadow protocol, using only single-qubit randomized Pauli measurements, to efficiently extract the AAI across extended subsystems. Through numerical simulations across diverse many-body paradigms, we show that the AAI can reveal distinct scrambling behaviors, resolving phenomena that range from dynamical confinement and ballistic transport to persistent scar revivals and many-body localization. This work highlights a pragmatic paradigm shift—from relying on optimal measurements to utilizing randomized local probes—for the characterization of complex quantum information dynamics.
\end{abstract}

\maketitle


\section{\label{sec:level1}Introduction}
The scrambling of quantum information throughout a many-body system stands as a central paradigm for our fundamental understanding of thermalization and the onset of quantum chaos~\cite{rigol2008thermalization, hosur2016chaos, deutsch1991quantum, srednicki1994chaos, maldacena2016bound}. For an isolated quantum system, an initially localized perturbation spreads irreversibly: information that is initially accessible via simple, few-body local observables is progressively diffused and "hidden" within highly entangled, non-local degrees of freedom across the entire system~\cite{hayden2007black, page1993average}. A challenge thus emerges for experimental characterization. While the true signature of scrambling resides in these global correlations, practical observations are inherently restricted to measuring local regions \cite{swingle2018unscrambling}. Tackling this locally-probed scrambling receives considerable attention, promising deep insights into phenomena ranging from the rapid scrambling in black holes \cite{sekino2008fast, shenker2014black, maldacena2016remarks} to the stark suppression of transport in condensed matter systems exhibiting many-body localization \cite{nandkishore2015many, abanin2019many, basko2006metal, pal2010many}.

A variety of powerful metrics have been developed to formally characterize the onset of information scrambling. Among the most prominent are out-of-time-order correlators (OTOCs) \cite{maldacena2016bound, swingle2018unscrambling}, which have emerged as a canonical measure for diagnosing quantum chaos by capturing the failure of initially commuting local operators to commute at later times. Complementary to OTOCs is the concept of operator size~\cite{von2018operator, qi2019quantum}, which directly quantifies the spatial support of a Heisenberg-evolved local operator as it becomes a complicated sum of non-local Pauli strings. The average of specific OTOCs over a complete basis is intimately connected to this operator size \cite{roberts2018operator}. Moreover, OTOCs are now becoming experimentally accessible in advanced quantum simulators \cite{li2017measuring, landsman2019verified, joshi2020quantum, mi2021information, garttner2017measuring, w1cp-l5vq, cg3f-rggs} and a protocol for measuring the operator size can be found in Ref.~\cite{PhysRevA.107.022407}. On the other hand, quantum information scrambling can be also directly revealed by information measure, such as the Holevo information \cite{holevo1973bounds, nielsen2010quantum} and others based on channel capacity~\cite{chen2025subsystem, Lo-Monaco-2023, Lo-Monaco_2024}. Remarkably, the Holevo information quantifies the maximum amount of classical information extractable from a quantum ensemble and has been proposed for exploring information scrambling of quantum many-body systems through local probes ~\cite{yuan2022quantum, sun2025revealing}. However, obtaining the Holevo information often relies on a state-dependent optimal measurement, which is usually hard to implement in experiments. 

An alternative paradigm that avoids optimal state-dependent measurements is to investigate the statistically typical information accessible via a set of randomized measurement bases. This philosophy underpins the concept of subentropy—a theoretical lower bound on accessible information for a density matrix~\cite{jozsa1994lower, bussandri2024renyi}, which finds applications from informational power of quantum measurements~\cite{dall2011informational} to deep thermalization \cite{mark2024maximum, ho2022exact, cotler2023emergent}.
Despite those advances, a question remains unexplored: can the information derived from averaging over randomized measurement bases serve as a sensitive probe to decode complex quantum information scrambling for disparate quantum many-body systems? Addressing this critical gap constitutes the primary motivation of our present work.

In this work, we address this gap by establishing an alternative framework for probing information scrambling based on the averaged accessible information~(AAI). By averaging the information gain over randomized measurement bases via the Haar random measurement, we derive a robust, purity-dependent metric for the AAI, denoted as $\chi_2$. Crucially, rather than demanding full state tomography to evaluate non-linear matrix logarithms, $\chi_2$ is fundamentally rooted in the evaluation of subsystem purities. This allows it to be efficiently and unbiasedly extracted using established protocols like classical shadows with randomized Pauli measurements \cite{huang2020predicting, elben2023randomized}. Building upon this experimental accessibility, we demonstrate that AAI can be a highly sensitive and robust probe of information dynamics of quantum systems. It can successfully resolve the chaotic confinement in the mixed-field Ising model, ballistic transport in the transverse-field Ising model, coherent revivals of quantum scars in the PXP model, and strict transport suppression in many-body localization. By replacing optimality with efficient randomized probes, this framework effectively bridges the gap between theoretical concepts of scrambling and practical experimental implementations on quantum simulators.

The remainder of this paper is organized as follows. In Sec.~\ref{THEORETICAL METRIC AND MEASUREMENT PROTOCOL}, we establish the theoretical framework by introducing the average accesssiable information and detail the classical shadows protocol used for its experimental estimation. In Sec.~\ref{result}, we apply this framework to four paradigmatic quantum many-body models and systematically analyze the resulting spatiotemporal information scrambling dynamics. Finally, we summarize our work in Sec.~\ref{conclusion}.

\section{THEORETICAL METRIC AND MEASUREMENT PROTOCOL}
\label{THEORETICAL METRIC AND MEASUREMENT PROTOCOL}
In this section, we formalize our framework for probing information scrambling. We first introduce the averaged accessible information ($\chi_2$), an operationally defined metric derived from the R\'enyi-2 entropy under randomized measurements, and subsequently detail the classical shadow protocol employed for its efficient experimental estimation.

\subsection{Averaged-Accessible Information}
\label{renyi-information}
In the study of quantum information scrambling, we track how an initially localized perturbation spreads through a quantum many-body system via observers are restricted to probing a local subsystem $A$. Consequently, the state of interest is not a pure state, but a reduced density matrix obtained by tracing out the environment. By preparing the system with different initial conditions—for instance, an unperturbed state $|\psi\rangle$ and a locally perturbed state $O|\psi\rangle$—the time evolution generates a quantum ensemble of mixed states, $\mathcal{E} = \{p_1=\frac{1}{2}, \rho_1(t); p_2=\frac{1}{2}, \rho_2(t)\}$. The survival of the initial local information is precisely determined by the distinguishability of these reduced density matrices.

In standard quantum information theory, the ultimate limit on state distinguishability is quantified by the Holevo information\cite{holevo1973bounds},
\begin{equation}
\chi(t) = S(\sum_i p_i \rho_i(t)) - \sum_i p_i S(\rho_i(t)),  
\end{equation}
where $S(\rho) = -\text{Tr}(\rho \log \rho)$ denotes the von Neumann entropy.
This quantity serves as the Optimal Accessible Information (OAI), representing the theoretical upper bound of information extractable by a receiver capable of performing an optimal, state-dependent measurement. To provide a more accessible lower bound, the subentropy $Q(\rho)$ was introduced \cite{jozsa1994lower}, giving rise to the analogous metric, 
\begin{equation}
\chi_Q(t) = Q(\sum_i p_i \rho_i(t)) - \sum_i p_i Q(\rho_i(t)). 
\end{equation}
Here, the subentropy $$Q(\rho) = - \sum_{k=1}^{n} \left( \lambda_k \ln \lambda_k \prod_{l \neq k} \frac{\lambda_k}{\lambda_k - \lambda_l} \right)$$ is a function defined explicitly in terms of the complete eigenspectrum of the density matrix (see Ref.~\cite{jozsa1994lower})

This necessitates a paradigm shift. Rather than characterizing scrambling through a theoretical upper bound (OAI) whose operational realization demands unattainable optimal measurements, we ask a more pragmatic question: how much information can be robustly and typically extracted if the observer probes the subsystem "blindly" using uninformed, randomized measurements? While subentropy $Q(\rho)$ serves as a lower bound derived via the standard logarithmic entropy, its practical utility is hindered by its reliance on the full state spectrum. We bridge this critical gap by formulating an analogous measure grounded instead in the R\'enyi-2 entropy, which naturally interfaces with randomized experimental protocols. We formalize the averaged accessible information (AAI)

\begin{align}
\chi_2 (t) = Q_2\left( \frac{\rho_1(t) + \rho_2(t)}{2} \right) - \frac{Q_2\left( \rho_1(t) \right) + Q_2\left( \rho_2(t) \right)}{2}.
\label{eq:chi}
\end{align}
By evaluating the information gain averaged over all possible Haar-random measurement bases (as detailed in Appendix~\ref{app:renyi2_derivation}), we derive the expression for the measure $Q_2$, 
\begin{equation}
Q_2(\rho) = \log[2/(1+\text{Tr}(\rho^2))].
\end{equation}
We begin by evaluating the classical R\'enyi-2 mutual information $I_2$ in the prototypical scenario of information extraction for an ensemble of pure states, $\mathcal{E}_{\text{pure}} = \{p_i, |\psi_i\rangle\}$. The average state of this ensemble is $\rho = \sum_i p_i |\psi_i\rangle\langle\psi_i|$. Suppose a receiver performs a projective measurement in an arbitrarily chosen basis $\mathcal{M} = \{|\alpha_j\rangle\}_{j=1}^d$. The conditional probability of outcome $j$ given the prepared state $|\psi_i\rangle$ is $P(j|i) = |\langle\alpha_j|\psi_i\rangle|^2$, and the marginal probability over the ensemble is $P(j) = \sum_i p_i P(j|i) = \langle\alpha_j|\rho|\alpha_j\rangle$. We quantify the information gained from this specific measurement,
\begin{align}
    I_2(\mathcal{E}_{\text{pure}} : \mathcal{M}) 
    = H_2\{P(j)\} - \sum_i p_i H_2\{P(j|i)\},
    \label{eq:app_renyi2_expanded}
\end{align}
where the classical R\'enyi-2 entropy $H_2(q) = -\log(\sum_j q_j^2)$ for the probability distribution $\{q_j\}$. 
To capture the typical information extracted via randomized measurements, we must average this quantity over all possible bases distributed according to the Haar measure.
Rather than computing the expectation of the logarithm, we operationally define our typical information metric by taking the logarithm of the Haar-averaged probability moments. This sequence of evaluating moments prior to the nonlinear logarithm directly aligns with the standard, robust experimental protocols for extracting R\'enyi entanglement entropies from randomized measurement data \cite{elben2018renyi}. The Haar-averaged mutual information becomes 

\begin{align} 
\label{eq:I2_annealed}
    \langle I_2 \rangle_{\text{Haar}} \equiv& -\log \left\langle \sum_j \langle \alpha_j | \rho | \alpha_j \rangle^2 \right\rangle_{\text{Haar}} \notag \\
    &+ \sum_i p_i \log \left\langle \sum_j |\langle \alpha_j | \psi_i \rangle|^4 \right\rangle_{\text{Haar}}.
\end{align}

For the first term involving the average state $\rho$, the Haar integration over a single basis state $|\alpha_j\rangle$ yields
\begin{align}
    \big\langle \langle \alpha_j | \rho | \alpha_j \rangle^2 \big\rangle_{\text{Haar}} = \frac{\operatorname{Tr}(\rho^2) + 1}{d(d+1)}.
    \label{eq:haar_rho}
\end{align}
For the conditional probability term, we evaluate the Haar average of the fourth moment of the overlap with a given pure state $|\psi_i\rangle$. Due to the unitary invariance of the Haar measure, this evaluates to
\begin{align}
    \left\langle \sum_j |\langle \alpha_j | \psi_i \rangle|^4 \right\rangle_{\text{Haar}} = \frac{2}{d+1}.
    \label{eq:haar_pure}
\end{align}

By substituting these Haar-averaged moments back into the R\'enyi-2 mutual information structure, we eliminate the basis dependence and arrive at a universal, operationally meaningful metric. In stark contrast to the Holevo bound, which dictates the maximum accessible information, this quantity represents the average information recovered via random probes, which we define as $Q_2$,
\begin{align}
    Q_2(\rho) \equiv\langle I_2 \rangle_{\text{Haar}}  &= \log\left( \frac{2}{1 + \sum_{k} \lambda_k^2} \right) \notag \\
    &= \log\left( \frac{2}{1 + \operatorname{Tr}(\rho^2)} \right).
    \label{eq:Q2_final}
\end{align}

This elegant algebraic expression for $Q_2$, dependent solely on the state's purity, unveils a remarkably mathematical structure. We rigorously prove that $Q_2(\rho)$ is equal to a complex spectral function involving all eigenvalues $\{\lambda_i\}_{i=1}^N$ of the density matrix,
\begin{equation}
    Q_2(\rho) = -\ln \left( \sum_{i=1}^N \frac{\lambda_i^{N+1}}{\prod_{\substack{j \neq i}} (\lambda_i - \lambda_j)} \right).
    \label{eq:Q2_spectral}
\end{equation}
This specific functional form, while previously identified in the context of generalized Wehrl entropies \cite{mintert2004wehrl}, is here derived from fundamental operational principles of randomized measurements. The full algebraic proof of this equivalence is detailed in Appendix~\ref{dengshi}. 

Crucially, this equivalence unveils an interesting connection to established quantum information bounds. As famously shown by Jozsa \cite{jozsa1994lower}, both the von Neumann entropy (the upper bound) and the subentropy (a theoretical lower bound) admit elegant representations as contour integrals over the complex plane. We reveal that $Q_2$ shares an intimate related contour integral structure, completing a trio of entropic measures. The explicit derivation of this elegant contour integral representation is provided in Appendix~\ref{integral}. Furthermore, we establish that $Q_2$ satisfies a suite of information-theoretic criteria: it is strictly positive, additive, and Schur-concave. Most remarkably, while the standard quantum R\'enyi-2 entropy $S_2(\rho) = -\log(\text{Tr}(\rho^2))$ famously violates concavity in general, our operationally derived $Q_2(\rho)$ restores this property. Specifically, for any density matrices $\rho_1, \rho_2$ and mixing parameter $\lambda \in [0,1]$, it obeys,
\begin{equation}
    Q_2\big(\lambda \rho_1 + (1-\lambda)\rho_2\big) \ge \lambda Q_2(\rho_1) + (1-\lambda) Q_2(\rho_2).
    \label{eq:Q2_concavity}
\end{equation}
The formal proof of the concavity condition is detailed in Appendix~\ref{app:concavity}. An immediate consequence of this concavity is that our defined dynamic metric $\chi_2(t) = Q_2\left(\sum_i p_i \rho_i(t)\right) - \sum_i p_i Q_2(\rho_i(t))$ is guaranteed to be non-negative ($\chi_2 \ge 0$). These properties demonstrate $Q_2$ as a mathematically rigorous and exceptionally well-behaved information measure.

\subsection{The Classical Shadows Protocol}
\label{classical shadow}
To map the spatiotemporal dynamics of information scrambling, we must evaluate $Q_2$ across numerous distinct local subsystems of a given size $L_A$ simultaneously. Traditionally, determining the required purities $\text{Tr}(\rho_A^2)$ for each possible subsystem requires a dedicated measurement per target. To circumvent this severe sampling bottleneck, we employ the classical shadows protocol \cite{huang2020predicting}. The advantage of this method is its ability to decouple the measurement phase from the target observables. A single set of randomized measurements on the entire $N$-qubit system is sufficient to efficiently and simultaneously predict many local properties, including the purities of all possible small subsystems.

The classical shadows protocol achieves this by constructing a description of the unknown quantum state $\rho$. The procedure involves executing $M$ independent experimental runs to generate a collection of "classical snapshots," denoted as $\{\hat{\rho}_i\}_{i=1}^M$. In each run, the state $\rho$ is first subjected to a randomly selected unitary operation $U$ from a predefined ensemble $\mathcal{U}$, followed by a projective measurement in the computational basis, yielding an outcome bit string $|\hat{b}_i\rangle$. Crucially, an unbiased state estimator, such that $\mathbb{E}[\hat{\rho}_i] = \rho$, is then reconstructed in classical post-processing by applying the inverse of the associated average measurement channel \cite{huang2020predicting}, defined as $\hat{\rho}_i = \mathcal{M}^{-1}(U_i^\dagger|\hat{b}_i\rangle\langle\hat{b}_i|U_i)$. Once this classical shadow is collected, the expectation value of any linear observable $O$ can be efficiently estimated by simply averaging its trace over the snapshots, $\langle O \rangle \approx \frac{1}{M}\sum_{i=1}^M \text{Tr}(O\hat{\rho}_i)$.

In this work, we employ the protocol based on random Pauli measurements. A significant advantage is that the inverse channel $\mathcal{M}^{-1}$ decomposes into a local tensor product, leading to a simple, explicit form for each snapshot,
\begin{align}
    \hat{\rho} = \bigotimes_{j=1}^N \left( 3 U_j^\dagger |\hat{b}_j\rangle\langle \hat{b}_j| U_j - \mathbb{I}_j \right).
\end{align}
The detailed derivation of this result can be found in Appendix~5C of Ref.~\cite{huang2020predicting}. The total number of experimental measurements $M$ required to estimate a collection of $N_{\text{obs}}$ different $k$-local observables to an additive precision $\epsilon$ is bounded by $M = \mathcal{O}\left(3^k \log(N_{\text{obs}}) / \epsilon^2\right)$ \cite{elben2023randomized, huang2020predicting}. This bound is entirely independent of the full system size, and can be used to evaluate $N_{\text{obs}} = \binom{L}{L_A}$ distinct $L_A$-local purities.

To further enhance experimental efficiency, we additionally utilize an efficient lookup table approach \cite{5-0.5--4}. This method imposes strict prerequisites. It is only applicable to standard single-qubit independent random Pauli measurements and the target observable must be a local quantity expressible as a product of single-qubit traces. The computational efficiency of the shadow kernel is rooted in a crucial simplification that the trace inner product between any two single-qubit snapshots, $\text{Tr}(\sigma_i^{(t)} \tilde{\sigma}_i^{(t)})$, is determined not by matrix multiplication but by a direct comparison of their corresponding measurement outcomes, namely the Pauli eigenstates $|s_i^{(t)}\rangle$ and $|\tilde{s}_i^{(t)}\rangle$. This comparison is quantified by the relation,

\begin{align}
    \text{Tr}(\sigma_i^{(t)} \tilde{\sigma}_i^{(t)}) = 9|\langle s_i^{(t)}|\tilde{s}_i^{(t)}\rangle|^2 - 4
\end{align}
which effectively creates a lookup table with only three possible values: 5~(when the measured eigenstates are identical), -4~(when they are orthogonal within the same Pauli basis), 1/2~(when they belong to different Pauli bases).This allows for the rapid calculation of the kernel by simply comparing the classical measurement data, bypassing complex matrix operations entirely.

To robustly estimate the expectation value of an observable, $o = \text{Tr}(O\rho)$, from a collection of independent classical snapshots $\{\hat{\rho}_i\}_{i=1}^M$, we employ the median-of-means\cite{huang2020predicting} estimation technique. While the standard empirical mean $\frac{1}{M}\sum_i \text{Tr}(O\hat{\rho}_i)$, provides an unbiased estimator, its convergence guarantees suffer from a poor dependence on the error probability $\delta$, with the required number of samples $M$ scaling as $\mathcal{O}(1/\delta)$. This scaling renders the estimator highly susceptible to corruption by a few statistical outliers. The median-of-means method mitigates this issue by partitioning the total $M = N \times K$ snapshots into $K$ disjoint batches, each of size $N$. For each batch, a separate empirical mean is computed, yielding $K$ independent estimators. The final prediction is then taken as the median of these $K$ means.

\begin{align}
\hat{o}_i(N,K) = 
&\mathrm{median}\left\{ \hat{o}_i^{(1)}(N,1), \dots, \hat{o}_i^{(K)}(N,1) \right\}
\quad,\\& \text{where}\quad
\hat{o}_i^{(k)} = \frac{1}{N} \sum_{j=N(k-1)+1}^{Nk} \mathrm{Tr}\!\left(O_i \hat{\rho}_j\right) \nonumber
\end{align}
The principal advantage of this procedure is its robustness against outliers. This method allows for highly accurate predictions from a finite number of measurements.

\section{SIMULATION RESULTS}
\label{result}
We now deploy th averaged accessible information $\chi_2$ to investigate non-equilibrium dynamics across four quantum many-body systems. We demonstrate the capacity of $\chi_2$ to unambiguously resolve distinct scrambling behavior.

\subsection{Models and Initial States}
\label{model}
In this paper, four systems are comparatively studied: the mixed-field Ising model (MFIM), the transverse-field Ising model (TFIM), the quantum many-body scar model (PXP), and the many-body localization model (MBL).

The mixed field Ising model (MFIM) incorporates longitudinal fields and boundary terms to mitigate finite-size effects,
\begin{align}
H_{\text{MFIM}} =& J \sum_{i=1}^{L-1} \sigma_i^z \sigma_{i+1}^z + g \sum_{i=1}^L \sigma_i^x + h \sum_{i=2}^{L-1} \sigma_i^z,
\end{align}
Following classic literature on mixed-field Ising model\cite{MFIM-Hamiton}, we fix $J=1$, set the transverse field strength to \(g=\frac{\sqrt{5}+5}{8}\) and the longitudinal field strength to \(h=\frac{\sqrt{5}+1}{4}\), to ensure that the energy scales $J$, $h$, and $g$ are comparable and mutually competing, thereby preventing any individual term (such as the Ising interaction or the magnetic field) from dominating the system dynamics. The parameters are constructed via $\sqrt{5}$ associated with the golden ratio, corresponding to a mathematically incommensurate configuration.

The transverse field Ising model (TFIM) is a canonical quantum spin system\cite{TFIM-bizu} exhibiting a quantum phase transition between ferromagnetic and paramagnetic phases. Its Hamiltonian is,
\[
H_{\text{TFIM}} = J \sum_{i=1}^{L-1} \sigma_i^z \sigma_{i+1}^z + g \sum_{i=1}^L \sigma_i^x,
\]
where \(J>0\) is the ferromagnetic Ising coupling, \(g=0.6\) the transverse field strength, \(\sigma_i^\alpha\) Pauli matrices, and \(L\) the system size. 

The PXP model is a constrained spin system featuring quantum many-body scars and non-thermal dynamics\cite{yuan2022quantum},
\[
H_{\text{PXP}} = \sum_i \hat{P}_{i-1} \sigma_i^x \hat{P}_{i+1}, \quad \hat{P}_j = \frac{1 - \sigma_j^z}{2},
\]
where \(\hat{P}_j\) projects out \(\sigma_j^z = +1\), constraining \(\sigma_i^x\) flips to neighbor sites with \(\sigma^z = -1\). This constraint yields quantum scars, causing persistent revivals in quantum dynamics.

The many-body localization (MBL) model describes a disordered interacting spin system that evades thermalization\cite{MBL-Hamiton},
\begin{align}
H_{\text{MBL}} &= J_\perp \sum_i \left(S_i^x S_{i+1}^x + S_i^y S_{i+1}^y\right) \nonumber \\&+ J_z \sum_i S_i^z S_{i+1}^z + \sum_i h_i S_i^z,
\end{align}
where \(J_\perp, J_z\) are anisotropic exchange couplings, \(S_i^\alpha = \sigma_i^\alpha/2\) spin-1/2 operators, and \(h_i\) quenched on-site disorder. In our numerical simulations, we adopt the Heisenberg MBL model with $J_\perp = J_z = 1$. We set the disorder strength $W = 8.0$, which places the system in the many-body localized phase for 1D spin chains. The on-site disorder fields $h_i$ are independently sampled from a uniform distribution over the interval $[-W, W]$, i.e., $h_i \sim \mathcal{U}[-8, 8]$. To ensure reproducibility of all numerical results presented in this work, we fix the random seed for disorder field generation to 42 across all experiments.

To probe the diverse scrambling mechanisms, we select the initial states to optimally excite the characteristic dynamics of each specific Hamiltonian. For the unconstrained spin chains—the MFIM and TFIM—we initiate the evolution from the fully polarized state, $|\Psi_0\rangle = |\uparrow\uparrow\cdots\uparrow\rangle$. This uniform state is ideally suited for tracking the propagation or confinement of a single local perturbation, providing a pristine visualization of the resulting information light cones. Conversely, for the PXP and MBL models, we exclusively employ the N\'eel configuration, $|\mathbb{Z}_2\rangle = |\uparrow\downarrow\uparrow\downarrow\cdots\uparrow\downarrow\rangle$. In the constrained PXP system, preparing this specific density-wave state is uniquely crucial for accessing the non-thermal subspace and sustaining the hallmark periodic revivals of quantum many-body scars. Simultaneously, in the MBL phase, the N\'eel state serves as a canonical probe, monitoring the extremely slow decay of its macroscopic staggered magnetization.

\subsection{Verification of averaged accessible information using Clifford circuit}

Our theoretical framework for the averaged accessible information ($\chi_2$), derived via Haar averaging in Appendix~\ref{app:renyi2_derivation}, relies on the continuous unitary measure. Mathematically, for quantities dependent on the second moments of the state (such as the purity $\operatorname{Tr}(\rho^2)$), averaging over any discrete set of unitary operations that constitutes at least a 2-design is equivalent to the continuous Haar average. To verify this equivalence and illustrate the operational feasibility of our metric without relying on the continuous Haar integral, we demonstrate that $\chi_2$ can be reliably extracted via finite random sampling from Clifford groups.

Following standard randomized measurement protocols, we sample unitaries from the global Clifford group, which forms a unitary 3-design. Specifically, in our numerical verification, for single-qubit subsystems ($L_A=1$), we uniformly sample unitaries directly from the complete set of 24 single-qubit Clifford elements. For two-qubit subsystems ($L_A=2$), we generate random Clifford unitaries by executing random quantum circuits composed of universal Clifford generators (Hadamard, Phase, and CNOT gates) with a depth of 50 layers, ensuring sufficient mixing to approximate a 2-design. By evaluating the empirical average of the squared measurement probabilities, we reconstruct the subsystem purities and subsequently $\chi_2$.

To evaluate the sampling convergence, we focus on the PXP model ($L=8$, $\Omega=1.0$) as a representative example. The system is initialized in the N\'eel state, $|\Psi_1\rangle = |\uparrow\downarrow\uparrow\downarrow\cdots\uparrow\downarrow\rangle$, and we track its distinguishability from a locally perturbed state $|\Psi_2\rangle$. The numerical results are presented in Fig.~\ref{fig:2-design(Clifford)}. At each time step $t$, we perform the Clifford sampling procedure for varying numbers of random unitaries ($N \in \{10, 50, 200\}$), and repeat this process over 5 independent trials for each $N$. As expected, for small sample sizes, the extracted $\chi_2$ exhibits noticeable statistical fluctuations. However, as the number of samples increases, the estimated dynamics rapidly converge to the theoretically exact values (solid black lines) computed via full state diagonalization.

This sampling analysis confirms that $\chi_2$ is not merely a theoretical construct dependent on infinite continuous integration, but a robust quantity that can be systematically and accurately approximated using standard, experimentally feasible Clifford 3-design sampling techniques.

\begin{figure}[htbp]
    \centering
    \includegraphics[width=0.5\textwidth]{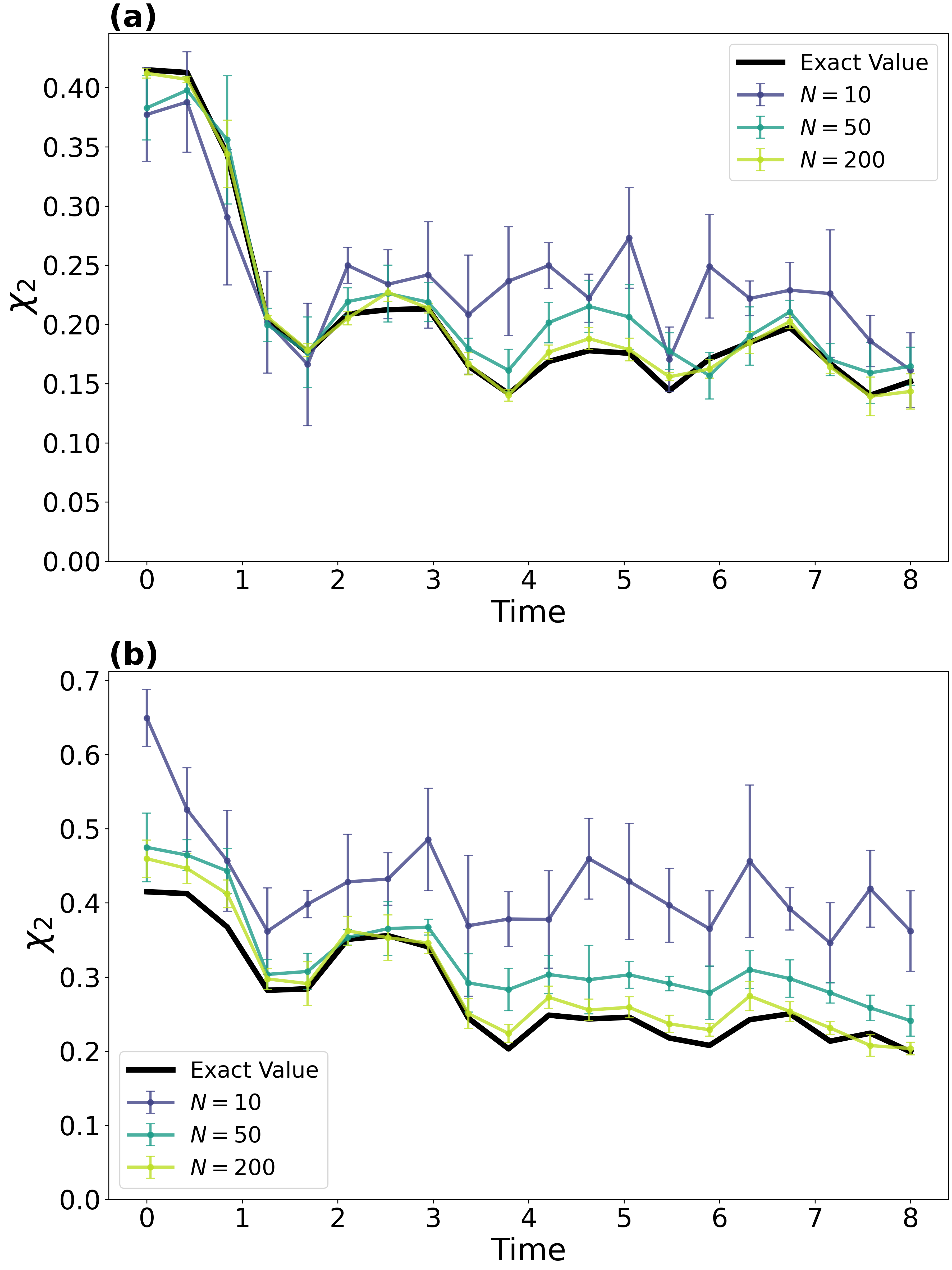}
    \caption{Convergence of AAI ($\chi_2$) via random Clifford sampling in the PXP model. Exact analytical evolution (solid black lines) is compared with statistical estimations for subsystem sizes (a) $L_A=1$ and (b) $L_A=2$. Estimations are obtained by sampling $N \in \{10, 50, 200\}$ random unitaries. For $L_A=1$, samples are drawn uniformly from the 24-element Clifford group; for $L_A=2$, samples are generated via 50-depth random Clifford circuits. Colored lines denote mean values, with error bars indicating the standard deviation across 5 independent trials.}
    \label{fig:2-design(Clifford)}
\end{figure}

\begin{figure*}[htbp]
    \centering
    \includegraphics[width=\textwidth]{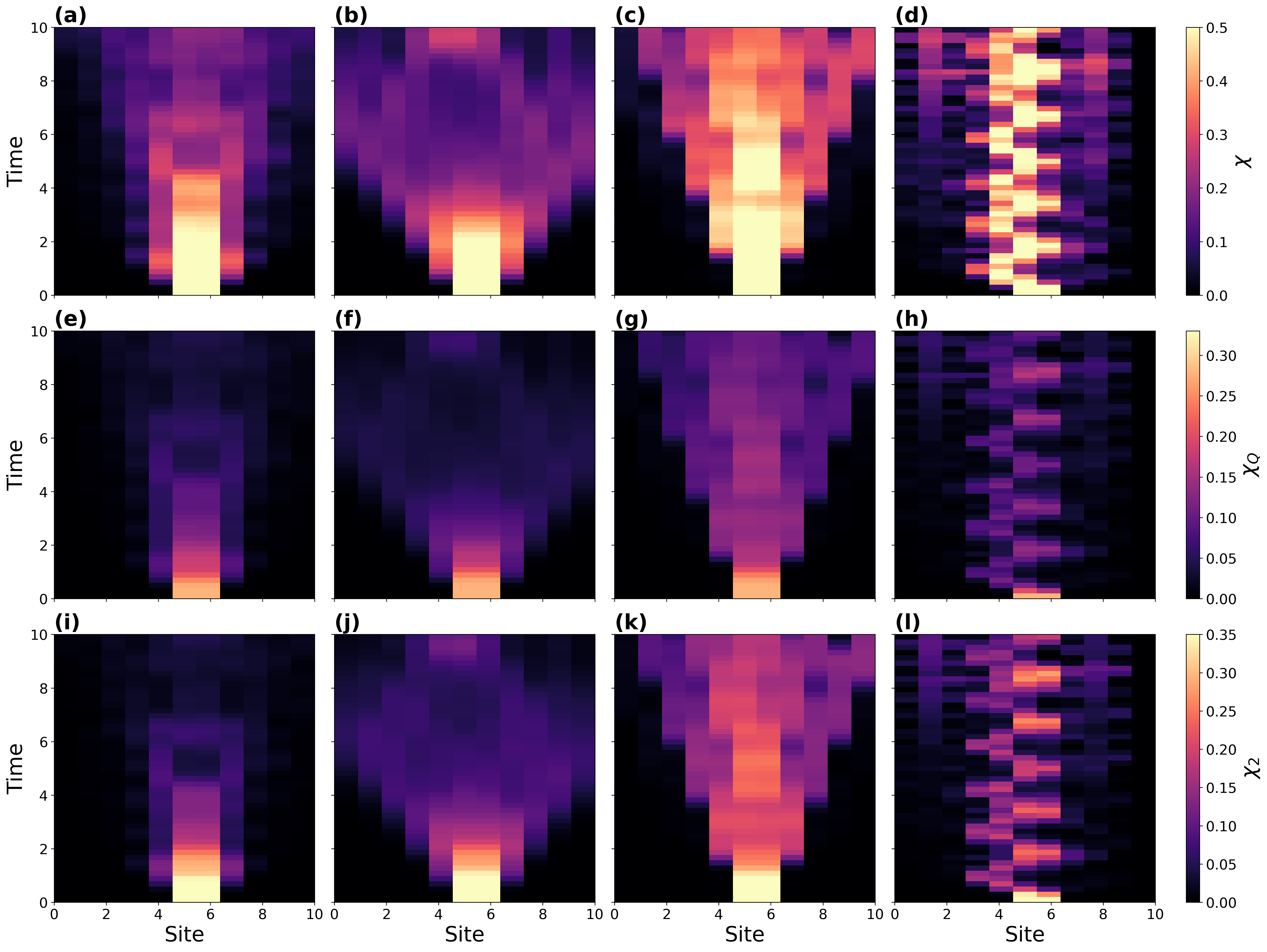} 
    \caption{Spatiotemporal resolution of information scrambling, $\chi_2$ versus Holevo $\chi$($L=12$). For each spatial site $x$ and time $t$, the plotted color intensity represents the respective information metric evaluated for a contiguous subsystem of size $L_A=2$ spanning sites $[x, x+1]$. The site index on the horizontal axis corresponds the left spin of the targeted subsystem. The top row (a)-(d) displays the standard Holevo information $\chi$, the middle row (e)-(h) displays the information from subentropy $\chi_Q$, while the bottom row (i)-(l) shows our averaged accessible information, $\chi_2$. (a), (e),(i) confinement in MFIM; (b), (f), (j) ballistic coherent transport in TFIM ($J=1, h=0.6$); (c), (g), (k) persistent scar revivals in PXP; and (d), (h), (l) strict localization in MBL.}
    \label{fig:heatmaps_spatiotemporal}
\end{figure*}

\subsection{Scrambling dynamics for different systems}
\label{results}

To initialize the information scrambling process, we construct an ensemble comprising the unperturbed base state $|\Psi_1\rangle$ and a locally perturbed state $|\Psi_2\rangle$. Specifically, the local perturbation is defined as a single-spin flip applied near the center of the chain, mathematically expressed as $|\Psi_2\rangle = \sigma_x^{(i)} |\Psi_1\rangle$, where $\sigma_x^{(i)}$ is the Pauli-X operator acting on site $i$. For the constrained PXP model, which forbids adjacent excitations due to the Rydberg blockade, we choose the perturbation to de-excite an active site (i.e., flipping $|\uparrow\rangle$ to $|\downarrow\rangle$). This protocol guarantees that the perturbed state $|\Psi_2\rangle$ remains confined within the physically valid scarred subspace. The subsequent decay of distinguishability between $|\Psi_1(t)\rangle$ and $|\Psi_2(t)\rangle$ serves as our signature of information scrambling.

\begin{figure*}[htbp]
    \centering
    \includegraphics[width=\textwidth]{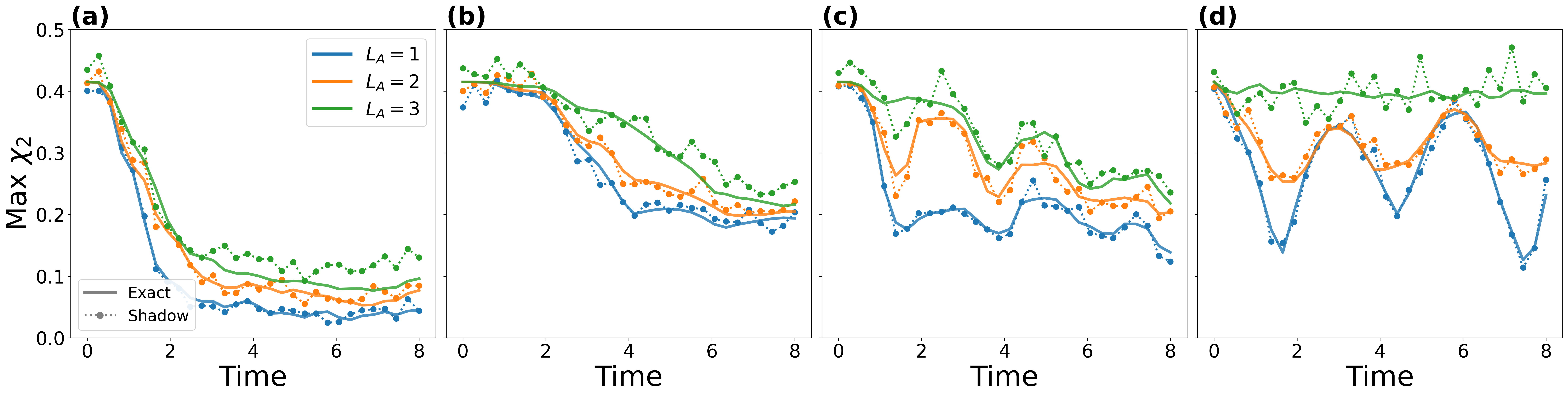} 
    \caption{Validation of the classical shadow protocol in non-ergodic regimes.
The row (a)-(d) compare the exact evolution of the maximal $\chi_2$ (solid lines) with estimations from the classical shadow protocol (dotted lines with markers, $N_{\text{shots}}=3000$) for various subsystem sizes $L_A \in \{1, 2, 3\}$. At each time step, the plotted maximum is evaluated over all $\binom{L}{L_A}$ possible spatial combinations of the $L_A$-site subsystem within the $L=10$ chain. Each column corresponds to a distinct dynamical paradigm: (a) the chaotic mixed-field Ising model; (b) the integrable transverse-field Ising model($J=1, h=0.3$); (c) the PXP model exhibiting quantum many-body scars and (d) the many-body localized system.}
    \label{fig:CS_four_sys_dynamic}
\end{figure*}

In Fig.~\ref{fig:heatmaps_spatiotemporal}, we provide a systematic comparison of the spatiotemporal dynamics of information scrambling across the four archetypal many-body paradigms. The top row (a)-(d) displays the evolution of the conventional Holevo information $\chi$, while the bottom row (i)-(l) shows our proposed averaged accessible information, $\chi_2$. The central finding is the striking qualitative equivalence between the two metrics. For each model, $\chi_2$ flawlessly reproduces the hallmark dynamical signatures captured by $\chi$: (a), (e), (i) the chaotic light cone and subsequent information scrambling in the MFIM; (b), (f), (j) the information scrambling of free quasiparticles in the integrable TFIM; (c), (g), (k) the persistent oscillatory revivals characteristic of quantum many-body scars in the PXP model; and (d), (h), (l) the information freezing and absence of transport in the MBL phase. This visual correspondence provides compelling evidence that $\chi_2$ serves as a high-fidelity, yet experimentally tractable, proxy for the full Holevo information, validating it as a powerful tool for probing non-equilibrium dynamics.

As shown in Fig.~\ref{fig:CS_four_sys_dynamic}, we can comprehensively understand the diverse dynamical behaviors of the averaged accessible information $\chi_2$ across the four investigated models through the lens of the Eigenstate Thermalization Hypothesis (ETH) and its various mechanisms of ergodicity breaking. To capture the leading edge of information propagation, the plotted $\chi_2$ represents the maximum value evaluated over all possible spatial combinations of the $L_A$-site subsystem at each time step. The robust retrieval of these diverse signatures is made possible by the classical shadow protocol. The dotted lines in Fig.~\ref{fig:CS_four_sys_dynamic} demonstrate that the shadow estimations track the dynamics (solid lines) across all subsystem sizes, verifying that our randomized measurement protocol successfully extracts the AAI even in complex many-body evolution.

Although the local information dynamics of the integrable TFIM and the chaotic MFIM appear phenomenologically similar---both exhibiting a rapid initial decay captured by the shadow estimations in Fig.~\ref{fig:CS_four_sys_dynamic}(a) and Fig.~\ref{fig:CS_four_sys_dynamic}(b)---their underlying physical mechanisms are distinct. In the TFIM, the apparent information loss is driven by the ballistic transport of unconstrained domain-wall excitations. Furthermore, as visualized in the spatiotemporal heatmaps [Fig.~\ref{fig:heatmaps_spatiotemporal}(f), (j)], the propagating information wavefront eventually reaches the ends of the open boundary chain and reflects inwards. Such boundary-induced reflections highlight the role of open boundary conditions in the dynamics: while local scrambling proceeds via ballistic propagation, the wavefront eventually reaches the system boundaries and is reflected back into the bulk. In a finite system of size $L$, this leads to interference between counter-propagating fronts and results in finite-size revivals at later times. In stark contrast, the longitudinal field in the MFIM explicitly breaks the $Z_2$ symmetry, inducing a linear confining potential (a ``string tension'') that binds the propagating domain walls. This confinement halts free propagation and forces the system into true many-body scrambling.

The PXP and many-body localization (MBL) models exemplify two distinct paradigms of ergodicity breaking in Fig.~\ref{fig:CS_four_sys_dynamic}(c) and Fig.~\ref{fig:CS_four_sys_dynamic}(d). The PXP model exhibits weak ergodicity breaking driven by quantum many-body scars. When quenched from the N\'eel state, the dynamics are trapped within a non-thermal subspace, severely delaying thermalization and giving rise to the pronounced, persistent oscillatory revivals of $\chi_2$. Conversely, the MBL system represents strong ergodicity breaking. The presence of strong disorder leads to the emergence of local integrals of motion, prohibiting energy transport and thermalization across the entire spectrum. Thus, the local information remains permanently localized. The classical shadow protocol captures this phenomenon, reflected by the persistent large-amplitude fluctuations characteristic of a single disorder realization.
The phenomenon of Many-Body Localization (MBL) is commonly studied in the one-dimensional disordered Heisenberg spin chain, governed by the Hamiltonian
\begin{align}
    H = &\sum_i \left[ J_{\perp} (\sigma_i^+ \sigma_{i+1}^- + \text{h.c.}) + J_z \sigma_i^z \sigma_{i+1}^z \right] + \sum_i h_i \sigma_i^z,
    \label{eq:mbl_hamiltonian}
\end{align}
where the on-site fields $h_i$ are drawn from a uniform random distribution $[-W, W]$. For sufficiently strong disorder $W$, the system enters an MBL phase, which, despite the presence of interactions, fails to thermalize.

\begin{figure}[htbp]
    \centering
     \includegraphics[width=0.5\textwidth]{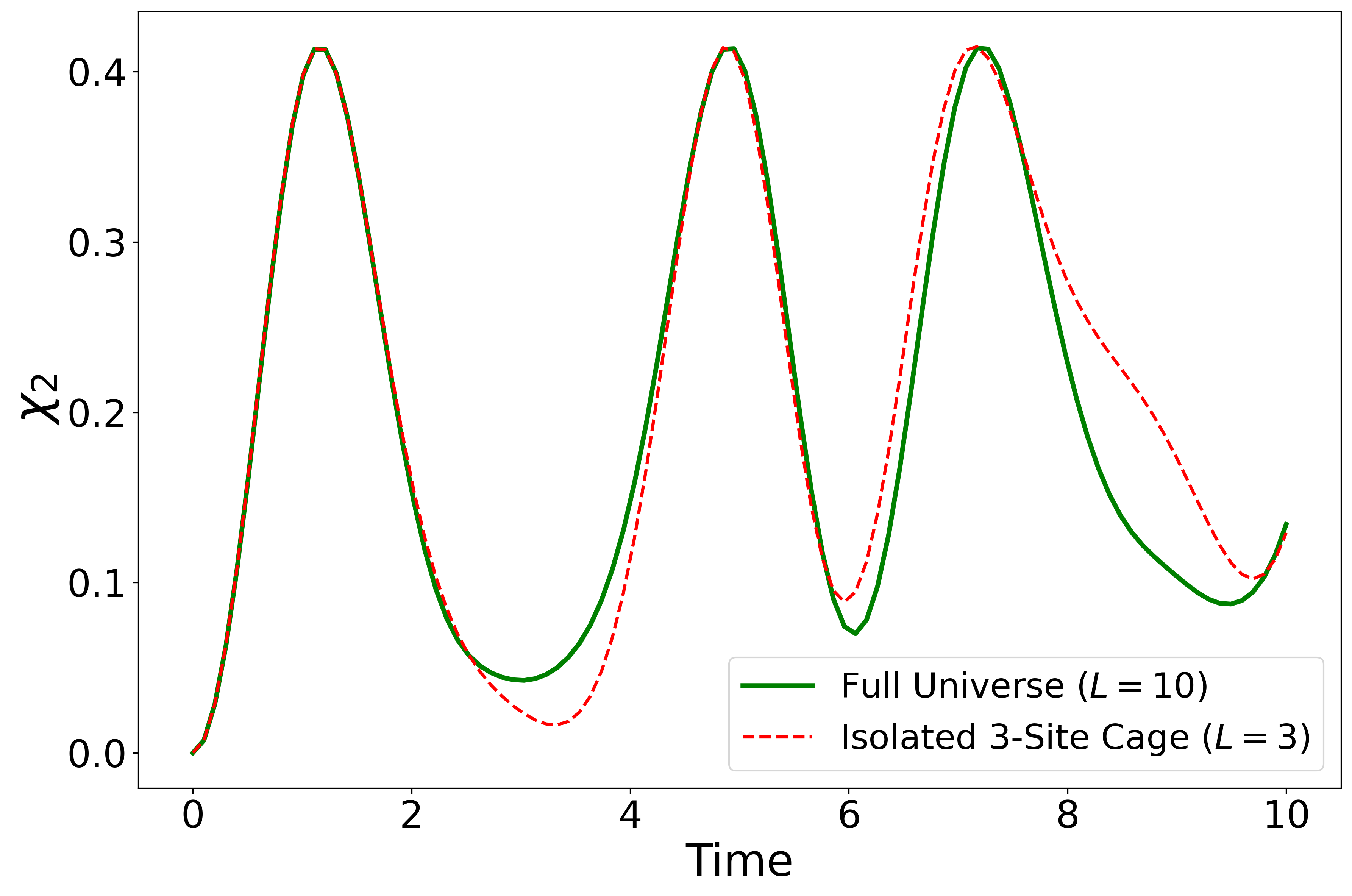}
    \caption{Dynamical confinement of information in the MBL phase. The solid green curve represents the dynamics within the full $H_{MBL}$ Hamiltonian system ($L=10$), while the dashed red curve corresponds to an artificially isolated three-site cage ($L=3$) sharing the same local disorder. The striking overlap of the high-amplitude oscillations at early-to-medium times confirms that the effective l-bits are strictly localized, rendering the local dynamics oblivious to the extended thermodynamic bath.
    }
    \label{fig:MBL single}
\end{figure}

A defining characteristic of the MBL phase is the emergence of a complete set of local integrals of motion, often termed "l-bits" $\{\tau_i^z\}$, which are quasi-local operators that commute with the Hamiltonian, $[{H}, \tau_i^z] = 0$ \cite{MBL-2014}. These l-bits are related to the physical spins ("p-bits") $\{\sigma_i^z\}$ via a quasi-local unitary transformation. This relationship implies that a physical spin operator is "dressed" by a cloud of l-bits, expressed as an expansion $\sigma_i^z = \sum_j Z_{ij} \tau_j^z + \sum_{j,k} A_{ijk} \tau_j^z \tau_k^z + \dots$, where the coefficients decay exponentially with distance, e.g., $|Z_{ij}| \propto e^{-|i-j|/\xi}$, for some localization length $\xi$.

In this l-bit basis, the system is described by an effective Hamiltonian with exponentially decaying interactions,
\begin{align}
    H_{\text{eff}} = &\sum_i \tilde{h}_i \tau_i^z + \sum_{i<j} \tilde{J}_{ij} \tau_i^z \tau_j^z + \dots, \\&\quad \text{where} \quad \tilde{J}_{ij} \propto e^{-|i-j|/\xi} \nonumber.
\end{align}

While the residual interactions $\tilde{J}_{ij}$ eventually lead to slow dephasing and logarithmic entanglement growth at asymptotically late times, the short-to-medium time dynamics are overwhelmingly dominated by the non-interacting Larmor precession in the local effective fields $\tilde{h}_i$. To explicitly demonstrate this stringent dynamical confinement, we compare the evolution of the averaged accessible information $\chi_2$, for a full chain ($L=10$) with a strictly truncated, isolated sub-region ($L=3$) under the same local disorder realization in Fig.~\ref{fig:MBL single}. The near-perfect coincidence of the oscillatory dynamics at early times verifies that the perturbation is strictly localized within an microscopic adjacent neighborhood. The thermodynamic bath of the extended system remains dynamically decoupled until the exponentially suppressed residual interactions ($\tilde{J}_{ij}$) finally induce a slight divergence at later times.

\section{DISCUSSION AND CONCLUSIONS}

It is instructive to comparatively analyze its operational advantages and intrinsic limitations relative to out-of-time-order correlators, the canonical metric for quantum scrambling. The two approaches possess complementary strengths. A principal advantage of OTOCs is their direct relation to operator growth and quantum chaos. In particular, OTOCs can reveal universal features such as ballistic operator spreading and exponential sensitivity to perturbations, quantities that are not directly accessible from $\chi_2$. Therefore, our approach should not be regarded as a replacement for OTOCs in studies where microscopic operator dynamics are the primary focus. The primary experimental advantage of our AAI framework, evaluated via the classical shadow protocol, lies in its reliance exclusively on forward time evolution and localized, single-qubit randomized Pauli measurements. Conversely, the direct experimental measurement of OTOCs typically necessitates the implementation of complex time-reversal protocols (i.e., evolving the system backward in time via $e^{-iHt}$) or the introduction of extensive auxiliary entangled modes \cite{swingle2018unscrambling, landsman2019verified}. In quantum simulators, imperfect control and environmental decoherence make precise backward evolution exceptionally challenging, rendering the forward-only shadow protocol pragmatic. We can view the AAI and OTOCs as complementary tools: the former provides an experimentally accessible, information-theoretic lens for local distinguishability, while the latter remains the standard for capturing  operator spreading. Establishing a more quantitative connection between AAI and operator-growth diagnostics, particularly the relation between $\chi_2$ and OTOCs in chaotic many-body systems, constitutes an interesting direction for future investigation.

In summary, we have introduced a paradigm of characterizing quantum information scrambling with uninformed randomized local probes. By rigorously averaging the information gain over all measurement bases via the Haar measure, we have derived the averaged accessible information (AAI), quantified by the R\'enyi-2 based metric $\chi_2$. Operationally, $\chi_2$ can be efficiently and unbiasedly estimated using the classical shadow protocol with single-qubit randomized Pauli measurements. Through extensive spatiotemporal tracking across four archetypal many-body systems, we have validated that the information captured by these "blind" local probes clearly and reliably reveals distinct thermodynamic dynamical behaviors—resolving chaotic confinement (MFIM), ballistic transport (TFIM), persistent scar revivals (PXP), and strict localization (MBL). Our framework has bridged the critical gap between abstract information-theoretic bounds and the experimental exploration of complex non-equilibrium dynamics in programmable quantum simulators.
\label{conclusion}

\section*{DATA AVAILABILITY}
The numerical simulation code and data used to generate the results in this study are openly available on GitHub~\cite{chen2026code}.

\begin{acknowledgments}
    
This work was supported by the National Natural Science Foundation of China (Grant No.12375013, No.12547109), the Guangdong Basic and Applied Basic Research Fund (Grant No.2023A1515011460), and Guangdong Provincial Quantum Science Strategic Initiative (Grant No. GDZX2503008).
    
\end{acknowledgments}
\appendix

\section{Analytical Derivation of $Q_2(\rho)$}
\label{app:renyi2_derivation}

In this appendix, we provide a detailed analytical derivation of the Haar-averaged classical R\'enyi-2 information $I_2$. We consider a quantum ensemble $\mathcal{E} = \{p_i, \ket{\psi_i}\}$ with the corresponding density matrix $\rho = \sum_i p_i \ket{\psi_i}\bra{\psi_i}$. Suppose we perform measurements in a randomly chosen basis $\mathcal{M} = \{\ket{\alpha_j}\}_{j=1}^d$. The conditional probability of obtaining the $j$-th outcome given the state $\ket{\psi_i}$ is $P(j|i) = |\langle \alpha_j | \psi_i \rangle|^2$, and the marginal probability is $P(j) = \sum_i p_i P(j|i) = \langle \alpha_j | \rho | \alpha_j \rangle$.

The classical R\'enyi-2 entropy for a probability distribution $q$ is defined as $H_2(q) = -\log \sum_j q_j^2$. Consequently, the classical R\'enyi-2 mutual information extracted from the measurement $\mathcal{M}$ is defined as,
\begin{align}
    I_2&(\mathcal{E} : \mathcal{M}) 
    = H_2(P(j)) - \sum_i p_i H_2(P(j|i)) \nonumber \\
    &= -\log\left(\sum_j P(j)^2\right) + \sum_i p_i \log\left(\sum_j P(j|i)^2\right)
    \label{eq:app_renyi2_expanded}
\end{align}

To evaluate the average over the Haar measure, we adopt the standard "annealed" approximation $\langle \log X \rangle_{\text{Haar}} \approx \log \langle X \rangle_{\text{Haar}}$. This procedure is justified in high-dimensional Hilbert spaces due to the concentration of measure, and it is the standard protocol for defining macroscopic Rényi entropies in randomized measurement settings \cite{mintert2004wehrl, elben2018renyi}. Under this established framework, the Haar-averaged information evaluates to,
\begin{align} \label{eq:I2_annealed}
    \langle I_2 \rangle_{\text{Haar}} =& -\log \left\langle \sum_j \langle \alpha_j | \rho | \alpha_j \rangle^2 \right\rangle_{\text{Haar}} \nonumber \\&+ \sum_i p_i \log \left\langle \sum_j |\langle \alpha_j | \psi_i \rangle|^4 \right\rangle_{\text{Haar}}.
\end{align}

We first evaluate the Haar average of the marginal probability squared. Let $x_k = |\langle \alpha_j | e_k \rangle|^2$. Due to the normalization of the basis vectors, we have $\sum_{k=1}^d x_k = 1$ and $x_k \ge 0$. We can rewrite the squared marginal probability as,
\begin{align} \label{eq:rho_expansion}
    \langle \alpha_j | \rho | \alpha_j \rangle^2 &= \left( \sum_k \lambda_k x_k \right)^2 \nonumber \\
    &= \sum_k \lambda_k^2 x_k^2 + 2 \sum_{k < l} \lambda_k \lambda_l x_k x_l.
\end{align}

Averaging over the Haar-random basis vector $\ket{\alpha_j}$ is equivalent to integrating the probabilities $x_k$ over the $(d-1)$-dimensional simplex $\Delta_{d-1}$ with a uniform measure. Utilizing the standard Dirichlet integral formula,
\begin{align}
    \int_{\Delta_{d-1}} \prod_{k=1}^d x_k^{\alpha_k - 1} dx_1 \dots dx_{d-1} = \frac{\prod_{k=1}^d \Gamma(\alpha_k)}{\Gamma\left(\sum_{k=1}^d \alpha_k\right)},
\end{align}
where $\Gamma(\cdot)$ is the Gamma function. The normalization factor for the uniform measure is $(d-1)!$. Using standard moments of the Dirichlet distribution ($\langle x_k^2 \rangle = \frac{2}{d(d+1)}$ and $\langle x_k x_l \rangle = \frac{1}{d(d+1)}$ for $k \neq l$), we can simplify Eq.~\eqref{eq:rho_expansion}. The Haar average for a single basis state $\ket{\alpha_j}$ yields,
\begin{align}
    \langle \langle \alpha_j | \rho | \alpha_j \rangle^2 \rangle_{\text{Haar}} &= \sum_k \lambda_k^2 \frac{2}{d(d+1)} + 2 \sum_{k < l} \lambda_k \lambda_l \frac{1}{d(d+1)} \nonumber \\
    &= \frac{1}{d(d+1)} \left[ \sum_k \lambda_k^2 + \left(\sum_k \lambda_k\right)^2 \right] \nonumber \\
    &= \frac{\text{Tr}(\rho^2) + 1}{d(d+1)},
\end{align}
where we used the identities $\sum_k \lambda_k = \text{Tr}(\rho) = 1$ and $\sum_k \lambda_k^2 = \text{Tr}(\rho^2)$.

Exploiting the unitary invariance, the sum over all $d$ basis vectors in $\mathcal{M}$ simply multiplies the result by $d$,
\begin{align}
    \left\langle \sum_j \langle \alpha_j | \rho | \alpha_j \rangle^2 \right\rangle_{\text{Haar}}&= \frac{\text{Tr}(\rho^2) + 1}{d+1}.
\end{align}
Therefore, the first term in Eq.~\eqref{eq:I2_annealed} evaluates to $-\log\left( \frac{\text{Tr}(\rho^2) + 1}{d+1} \right)$.

For the second term, we evaluate the Haar average of $\sum_j |\langle \alpha_j | \psi_i \rangle|^4$ for a given pure state $\ket{\psi_i}$. By letting $y_{j,i} = |\langle \alpha_j | \psi_i \rangle|^2$, we notice that $\sum_j y_{j,i} = 1$. Due to Haar symmetry, the average is independent of the choice of the initial pure state $\ket{\psi_i}$ or the specific basis vector $j$, leading to $\langle \sum_j y_{j,i}^2 \rangle_{\text{Haar}} = d \langle y_{1,i}^2 \rangle_{\text{Haar}}$.

Using the same integration results derived above for pure states (where only one eigenvalue is $1$ and the rest are $0$), we find $\langle y_{1,i}^2 \rangle_{\text{Haar}} = \frac{2}{d(d+1)}$. Thus,
\begin{align}
    \left\langle \sum_j |\langle \alpha_j | \psi_i \rangle|^4 \right\rangle_{\text{Haar}} = d \cdot \frac{2}{d(d+1)} = \frac{2}{d+1}.
\end{align}
Since $\sum_i p_i = 1$, the second term in Eq.~\eqref{eq:I2_annealed} evaluates to $\log\left( \frac{2}{d+1} \right)$.

Combining the evaluations of both terms, the Haar-averaged classical R\'enyi-2 mutual information reduces to a simple and elegant form,
\begin{align} \label{eq:I2_final}
    Q_2(\rho)\equiv\langle I_2 \rangle_{\text{Haar}} &= -\log\left( \frac{\text{Tr}(\rho^2) + 1}{d+1} \right) + \log\left( \frac{2}{d+1} \right) \nonumber \\
    &= \log\left( \frac{2}{1 + \text{Tr}(\rho^2)} \right).
\end{align}
Notice that the system dimension $d = 2^N$ cancels out completely in the final expression. This analytical result rigorously links the classical information to the R\'enyi-2 entropy $H_2(\rho)$, establishing a solid foundation for probing information dynamics.

\section{Proof of the Logarithmic Purity Identity}
\label{dengshi}

\begin{lemma}
Let $\rho$ be an $n$-dimensional  density matrix with eigenvalues $\{\lambda_i\}_{i=1}^n$. The following algebraic identity holds universally,
\begin{align}
    -\ln \left( \sum_{i=1}^n \frac{\lambda_i^{n+1}}{\prod_{j \neq i} (\lambda_i - \lambda_j)} \right) = \ln \left( \frac{2}{1 + \operatorname{Tr}(\rho^2)} \right).
\end{align}
\end{lemma}

\begin{proof}
We first assume that the density matrix $\rho$ has a non-degenerate spectrum, i.e., $\lambda_i \neq \lambda_j$ for all $i \neq j$. Consider the characteristic polynomial of $\rho$, defined as $g(x) = \prod_{i=1}^n (x - \lambda_i)$. Expanding $g(x)$ yields,
\begin{align}
    g(x) = x^n - S_1 x^{n-1} + S_2 x^{n-2} - \dots + (-1)^n S_n,
\end{align}
where $S_k$ are the elementary symmetric polynomials of the eigenvalues. Specifically, $S_1 = \sum_i \lambda_i = \operatorname{Tr}(\rho) = 1$, and $S_2 = \sum_{i < j} \lambda_i \lambda_j$.

Performing division of $x^{n+1}$ by $g(x)$, there exist unique polynomials $q(x)$ and $r(x)$ such that $x^{n+1} = q(x)g(x) + r(x)$, with the degree of the remainder $\deg(r) \le n-1$. 
Evaluating this align at the roots of $g(x)$, we obtain $r(\lambda_i) = \lambda_i^{n+1}$ since $g(\lambda_i) = 0$. Using the Lagrange interpolation theorem, the remainder $r(x)$ can be uniquely constructed as,
\begin{align}
    r(x) = \sum_{i=1}^n \lambda_i^{n+1} \prod_{j \neq i} \frac{x - \lambda_j}{\lambda_i - \lambda_j}. \label{eq:lagrange}
\end{align}
The coefficient of the highest-order term $x^{n-1}$ in $r(x)$ is the argument of the logarithm in the left-hand side of Eq.~(1), which corresponds to the complete homogeneous symmetric polynomial of degree two, $h_2(\lambda_1, \dots, \lambda_n)$.

Alternatively, one can determine this coefficient via direct algebraic expansion. Multiplying the characteristic polynomial by $x$, we have
\begin{align}
    x^{n+1} &= x \cdot g(x) + S_1 x^n - S_2 x^{n-1} + \dots \nonumber \\
            &= x \cdot g(x) + S_1 \left( g(x) + S_1 x^{n-1} - \dots \right) \nonumber \\
            &- S_2 x^{n-1} + \dots \nonumber \\
            &= (x + S_1)g(x) + (S_1^2 - S_2)x^{n-1} + \mathcal{O}(x^{n-2}).
\end{align}
By the uniqueness of the polynomial division, the coefficient of $x^{n-1}$ in the remainder $r(x)$ must equal $S_1^2 - S_2$. Equating this to the coefficient obtained from Eq.~\eqref{eq:lagrange}, we arrive at the identity,
\begin{align}
    \sum_{i=1}^n \frac{\lambda_i^{n+1}}{\prod_{j \neq i} (\lambda_i - \lambda_j)} = S_1^2 - S_2.
\end{align}
Using the properties of the trace, we recall that $S_1 = \operatorname{Tr}(\rho) = 1$ and $S_1^2 = \sum_i \lambda_i^2 + 2\sum_{i < j} \lambda_i \lambda_j = \operatorname{Tr}(\rho^2) + 2S_2$. Rearranging this gives $S_2 = \frac{1}{2}[1 - \operatorname{Tr}(\rho^2)]$. Substituting $S_1$ and $S_2$ back into the coefficient yields,
\begin{align}
    S_1^2 - S_2 = 1 - \frac{1 - \operatorname{Tr}(\rho^2)}{2} = \frac{1 + \operatorname{Tr}(\rho^2)}{2}.
\end{align}
Taking the negative logarithm of both sides completes the algebraic proof for the non-degenerate case. 
\end{proof}

\section{Contour Integral Representation of the averaged accessible information}
\label{integral}

In this appendix, we establish the formal equivalence between the contour integral representation of the averaged accessible information, its complex spectral expansion, and its final simplified form based on purity. 

Inspired by the elegant contour integral formulations of the von Neumann entropy $S(\rho)$ and the subentropy $Q(\rho)$ introduced by Jozsa \textit{et al.} \cite{jozsa1994lower}, we define the core quantity of our randomized measurement protocol via an analogous integral over the complex plane,
\begin{equation}
    \mathcal{J}(\rho) \equiv \frac{1}{2\pi i} \oint_\Gamma z \det(I - \rho/z)^{-1} dz,
    \label{eq:app_integral_def}
\end{equation}
where the contour $\Gamma$ is a simple closed curve traversed counterclockwise, enclosing all the non-zero eigenvalues $\{\lambda_i\}_{i=1}^N$ of the $N$-dimensional density matrix $\rho$. 

To evaluate this integral, we first express the determinant in terms of the eigenvalues of $\rho$,
\begin{equation}
    \det(I - \rho/z)^{-1} = \prod_{k=1}^N \left( 1 - \frac{\lambda_k}{z} \right)^{-1} = z^N \prod_{k=1}^N \frac{1}{z - \lambda_k}.
\end{equation}
Substituting this back into Eq.~\eqref{eq:app_integral_def}, the integral takes the form of a rational function,
\begin{equation}
    \mathcal{J}(\rho) = \frac{1}{2\pi i} \oint_\Gamma \frac{z^{N+1}}{\prod_{k=1}^N (z - \lambda_k)} dz.
    \label{eq:app_integral_rational}
\end{equation}

Assuming, for the moment, that the density matrix $\rho$ possesses a non-degenerate spectrum (i.e., $\lambda_i \neq \lambda_j$ for all $i \neq j$), the integrand exhibits exactly $N$ simple poles located at $z = \lambda_i$ within the contour $\Gamma$. According to Cauchy's Residue Theorem, the integral evaluates to the sum of the residues at these poles,
\begin{align}
    \mathcal{J}(\rho) &= \sum_{i=1}^N \operatorname{Res}_{z=\lambda_i} \left[ \frac{z^{N+1}}{\prod_{k=1}^N (z - \lambda_k)} \right] \notag \\
    &= \sum_{i=1}^N \lim_{z \to \lambda_i} (z - \lambda_i) \frac{z^{N+1}}{\prod_{k=1}^N (z - \lambda_k)} \notag \\
    &= \sum_{i=1}^N \frac{\lambda_i^{N+1}}{\prod_{\substack{j=1 \\ j \neq i}}^N (\lambda_i - \lambda_j)}.
    \label{eq:app_spectral_form}
\end{align}
This result establishes the equivalence between the contour integral and the intricate spectral formulation. While derived under the assumption of non-degeneracy, this identity holds universally for all physical density matrices. Since both sides of Eq.~\eqref{eq:app_spectral_form} are continuous functions of the matrix elements of $\rho$ (the apparent singularities on the right-hand side are removable), the result extends to degenerate spectra via a standard continuity argument or L'H\^opital's rule.

Finally, we connect this spectral form to the operational purity $\operatorname{Tr}(\rho^2)$. As detailed in the main text (and fully proven via polynomial long division in Ref.~\cite{mendonca2008wehrl}), the right-hand side of Eq.~\eqref{eq:app_spectral_form} algebraically simplifies to $[1 + \operatorname{Tr}(\rho^2)]/2$. 

Therefore, by defining our averaged accessible information as $Q_2(\rho) \equiv -\ln \mathcal{J}(\rho)$, we obtain the complete chain of equivalences central to our framework,
\begin{align}
    Q_2(\rho) &= -\ln \left[ \frac{1}{2\pi i} \oint_\Gamma z \det(I - \rho/z)^{-1} dz \right] \notag \\
    &= -\ln \left( \sum_{i=1}^N \frac{\lambda_i^{N+1}}{\prod_{\substack{j \neq i}}(\lambda_i - \lambda_j)} \right) \notag \\
    &= \ln \left( \frac{2}{1 + \operatorname{Tr}(\rho^2)} \right).
    \label{eq:app_Q2_chain}
\end{align}
This cohesive mathematical structure bridges the abstract, optimal-measurement-bounded entropy theories with our pragmatic, randomized-probe observables.

\section{Rigorous Proof of Strict Concavity for the averaged accessible Information}
\label{app:concavity}

A requirement for any well-behaved entropic measure is concavity. It dictates that classical mixing of quantum states must not decrease our overall ignorance. While the standard quantum R\'enyi-2 entropy, $S_2(\rho) = -\log(\text{Tr}(\rho^2))$, is known to violate concavity in general, we rigorously prove below that our operationally derived metric, the averaged accessible Information $Q_2(\rho)$, perfectly cures this pathology and exhibits strict concavity globally across the state space.

We define the metric as $Q_2(\rho) = -\log(f(\rho))$, where $f(\rho) = \frac{1 + \mathcal{P}(\rho)}{2}$ and $\mathcal{P}(\rho) = \text{Tr}(\rho^2)$ is the purity.

To prove strict concavity on the convex set of density matrices $\mathcal{D}(\mathcal{H})$, it is necessary and sufficient to show that the second derivative of $Q_2$ along any valid straight line segment connecting two distinct density matrices is negative.

Let $\rho_0, \rho_1 \in \mathcal{D}(\mathcal{H})$ be two distinct density matrices ($\rho_0 \neq \rho_1$). We define a parameterized convex path for $t \in [0, 1]$,
\begin{equation}
    \rho(t) = (1-t)\rho_0 + t\rho_1 = \rho_0 + t\Delta,
\end{equation}
where $\Delta = \rho_1 - \rho_0$. By definition, $\Delta$ is a non-zero, traceless Hermitian matrix ($\text{Tr}(\Delta) = 0$).

We define the scalar function $G(t) \equiv Q_2(\rho(t))$. Our goal is to prove $G''(t) < 0$ for all $t \in (0, 1)$. Without loss of generality, it suffices to evaluate the second derivative at $t=0$ for an arbitrary base state $\rho_0$ and any allowable perturbation $\Delta$.

The purity along the path is $P(t) = \text{Tr}[(\rho_0 + t\Delta)^2]$. Its derivatives at $t=0$ are,
\begin{align}
    P'(0) &= 2\text{Tr}(\rho_0\Delta), \label{eq:app_P_prime}\\
    P''(0) &= 2\text{Tr}(\Delta^2). \label{eq:app_P_dprime}
\end{align}
Since $\Delta$ is non-zero and Hermitian, $\text{Tr}(\Delta^2) > 0$.

Evaluating the derivatives of $G(t) = -\log\left( \frac{1 + P(t)}{2} \right)$ via the chain rule yields,
\begin{equation}
    G''(0) = -\frac{P''(0)\big(1 + P(0)\big) - \big(P'(0)\big)^2}{\big(1 + P(0)\big)^2}.
\end{equation}
To prove strict concavity ($G''(0) < 0$), we must show that the numerator, which we define as $K$, is strictly positive,
\begin{equation}
    K \equiv 2\text{Tr}(\Delta^2)\big(1 + \text{Tr}(\rho_0^2)\big) - 4\big(\text{Tr}(\rho_0\Delta)\big)^2 > 0.
    \label{eq:app_hessian_cond}
\end{equation}

We invoke the Cauchy-Schwarz inequality for the Hilbert-Schmidt inner product, $\big(\text{Tr}(AB)\big)^2 \le \text{Tr}(A^2)\text{Tr}(B^2)$ for Hermitian $A, B$. Applying this to the second term,
\begin{equation}
    4\big(\text{Tr}(\rho_0\Delta)\big)^2 \le 4\text{Tr}(\rho_0^2)\text{Tr}(\Delta^2).
    \label{eq:app_CS}
\end{equation}
Substituting this upper bound into Eq.~\eqref{eq:app_hessian_cond}, we establish a lower bound for $K$,
\begin{align}
    K &\ge 2\text{Tr}(\Delta^2)\big(1 + \text{Tr}(\rho_0^2)\big) - 4\text{Tr}(\rho_0^2)\text{Tr}(\Delta^2) \notag \\
    &= 2\text{Tr}(\Delta^2) \big[ 1 - \text{Tr}(\rho_0^2) \big].
    \label{eq:app_final_bound}
\end{align}

Since $\rho_0 \in \mathcal{D}(\mathcal{H})$, its purity satisfies $\text{Tr}(\rho_0^2) \le 1$, we analyze the two possible cases:

Case 1: $\rho_0$ is a mixed state ($\text{Tr}(\rho_0^2) < 1$).
In this case, $[1 - \text{Tr}(\rho_0^2)] > 0$. Since $\text{Tr}(\Delta^2) > 0$, Eq.~\eqref{eq:app_final_bound} guarantees that $K > 0$, confirming strict concavity.

Case 2: $\rho_0$ is a pure state ($\text{Tr}(\rho_0^2) = 1$).
Here, the lower bound in Eq.~\eqref{eq:app_final_bound} evaluates to $0$. However, we must examine the inequality in Eq.~\eqref{eq:app_CS} more closely. For a pure state, $\rho_0^2 = \rho_0$. The term we are bounding is $4\big(\text{Tr}(\rho_0\Delta)\big)^2$. 

Recall that $\Delta = \rho_1 - \rho_0$, where $\rho_1$ is another valid density matrix.
\begin{align}
    \text{Tr}(\rho_0\Delta) &= \text{Tr}(\rho_0\rho_1) - \text{Tr}(\rho_0^2) = \text{Tr}(\rho_0\rho_1) - 1.
\end{align}
Because $\rho_0$ and $\rho_1$ are distinct, and $\rho_0$ is pure, their overlap must be less than 1, $\text{Tr}(\rho_0\rho_1) < 1$. Thus, $\text{Tr}(\rho_0\Delta) < 0$.

Now, let us examine the Cauchy-Schwarz bound Eq.~\eqref{eq:app_CS} specifically for pure states. The equality holds if and only if $\Delta = c \rho_0$ for some scalar $c$. Taking the trace of both sides yields $\text{Tr}(\Delta) = c \text{Tr}(\rho_0)$. Since $\Delta$ is traceless and $\text{Tr}(\rho_0)=1$, we must have $c=0$, which implies $\Delta = 0$. But by definition, $\rho_1 \neq \rho_0$, so $\Delta \neq 0$.

Therefore, the Cauchy-Schwarz inequality in Eq.~\eqref{eq:app_CS} must be less than:
\begin{equation}
    4\big(\text{Tr}(\rho_0\Delta)\big)^2 < 4\text{Tr}(\rho_0^2)\text{Tr}(\Delta^2).
\end{equation}
Because the inequality is strict, the original expression $K$ in Eq.~\eqref{eq:app_hessian_cond} evaluates to,
\begin{align}
    K &> 2\text{Tr}(\Delta^2)\big(1 + 1\big) - 4\text{Tr}(\rho_0^2)\text{Tr}(\Delta^2) \notag \\
      &= 4\text{Tr}(\Delta^2) - 4(1)\text{Tr}(\Delta^2) = 0.
\end{align}
Thus, even at the boundary of pure states, $K$ is strictly positive, ensuring $G''(0) < 0$.

In conclusion, for any valid path within the density matrix space, the second derivative of $Q_2(\rho)$ is negative. The constant $+1$ in the argument of $Q_2(\rho)$ acts as a critical mathematical stabilizer, suppressing the non-concave cross-terms. This establishes $Q_2(\rho)$ as a strictly concave and mathematically rigorous measure of quantum ignorance.


%

\end{document}